\documentclass[leqno]{article}
\usepackage{amssymb,amsmath,amsthm}
\usepackage{verbatim}

\author{\textsc{M. Castrill\'on L\'opez} \\
Instituto de Ciencias Matem\'aticas CSIC-UAM-UC3M-UCM \\
Departamento de Geometr\'{\i}a y Topolog\'{\i}a \\
Facultad de Matem\'aticas,\\
Universidad Complutense de Madrid, 28040-Madrid, Spain \\
\emph{E-mail:\/} \texttt{mcastri@mat.ucm.es}
\and
\textsc{J. Mu\~{n}oz Masqu\'e} \\
Instituto de F\'{\i}sica Aplicada, CSIC \\
C/ Serrano 144, 28006-Madrid, Spain \\
\emph{E-mail:\/} \texttt{jaime@iec.csic.es}}

\title{\textbf{Hamiltonian Structure
of Gauge-Invariant Variational Problems}}

\newtheorem{theorem}{Theorem}[section]
\newtheorem{proposition}[theorem]{Proposition}
\newtheorem{lemma}[theorem]{Lemma}
\newtheorem{corollary}[theorem]{Corollary}

\theoremstyle{remark}
\newtheorem{remark}[theorem]{Remark}
\newtheorem{definition}[theorem]{Definition}

\begin{document}

\date{}
\maketitle

\begin{abstract}
\noindent Let $C\to M$ be the bundle
of connections of a principal bundle
on $M$. The solutions to Hamilton-Cartan
equations for a gauge-invariant Lagrangian
density $\Lambda $ on $C$ satisfying
a weak condition of regularity, are shown
to admit an affine fibre-bundle structure
over the set of solutions to Euler-Lagrange
equations for $\Lambda $. This structure
is also studied for the Jacobi fields
and for the moduli space of extremals.
\end{abstract}

\bigskip


\medskip


\medskip


\medskip

\noindent \emph{Acknowledgments:\/}
Supported by Ministerio de Ciencia
e Innovaci\'on of Spain under grants
\#MTM2008--01386 and \#MTM2007--60017.

\section{Introduction}

Let $p\colon E\to M$ be a fibred manifold
and let $\Lambda $ be a Lagrangian density
on $J^1E$. The solutions to the Hamilton-Cartan
(H-C, in short) equations (e.g., see \cite{Garcia0},
\cite{GS}, or \cite{MP}) for $\Lambda $, are
the sections $\bar{s}\colon M\to J^1E$
of the canonical projection $p_1\colon J^1E\to M$
fulfilling the equation \eqref{HC} below.
If $\Lambda $ is regular, i.e., its Hessian metric
is non-singular (cf.\ Definition \ref{regular}
below), then every solution $\bar{s}$
to the Hamilton-Cartan equations
is holonomic, i.e., $\bar{s}=j^1s $, and $s$
is a critical section for $\Lambda $, that is,
a solution to the Euler-Lagrange (E-L, in short)
equations. Nevertheless, a Lagrangian defined
on the bundle of connections $p\colon C\to M$
of a principal bundle $\pi\colon P\to M$,
which is invariant under the gauge group of $P$,
is never regular. Therefore, Hamilton-Cartan
equations and Euler-Lagrange are not equivalent
for all the gauge invariant variational problems
on connections.

In the present work we show that the solutions
to H-C equations for a gauge-invariant Lagrangian
density $\Lambda $ on the bundle $C$ satisfying
a weak condition of regularity (most of interesting
gauge-invariant Lagrangians in the field theory
satifies this condition) admits an affine fibre-bundle
structure over the set of critical sections
of the variational problem defined by the density
under consideration. Moreover, the structure
of this fibration is completely determined; see
Theorem \ref{th1} below. This proves, in particular,
that the E-L equations of a gauge-invariant Lagrangian
are essentially of first order.

Such a structure is meaningful even from the point
of view of the observables of the field theory,
because for every extremal section $s$ the curvature
is constant along the fibre over $j^1s$.

By passing to the quotient such a fibre bundle
modulo the gauge group, we obtain
a---not necessary trivial---vector bundle associated
to the principal bundle of the moduli of extremals
of $\Lambda $; see Proposition \ref{moduli} below.
Finally, the aforementioned affine structure is also
studied for the Jacobi fields; i.e., for the vector
fields in the kernel of the linearization of the E-L
and H-C operators at any extremal (cf.\ Theorem \ref{T2}
below).

\section{Notations and preliminaries}

\subsection{Jet bundles\label{JetBundles}}

The bundle of $r$-jets of a fibred manifold
$p\colon E\to M$ is denoted by
$p_r\colon J^r(p)=J^rE\to M$ with projections
$p_{rk}\colon J^rE\to J^kE$, $r\geq k$,
the $r$-jet extension of a section
$s\colon M\to E$ of $p$ being denoted by $j^rs$.
Every $p$-fibred coordinate system
$(x^i,y^\alpha )$, $1\leq i\leq n=\dim M$,
$1\leq \alpha \leq m=\dim E-n$,
defined on an open subset $U\subseteq E$,
induces a coordinate system $(x^i,y_I^\alpha )$,
$0\leq|I|\leq r$, on $(p_{r0})^{-1}(U)=J^rU$;
namely, $y_I^\alpha (j_x^rs)=(\partial ^{|I|}
(y^\alpha \circ s)/\partial x^I)(x)$,
with $y_0^\alpha =y^\alpha $.
Every fibred map $\Phi \colon E\to E^\prime $
whose induced map $\phi \colon M\to M^\prime $
on the base manifold is a diffeomorphism,
induces a map $\Phi ^{(r)}\colon J^rE
\to J^rE^\prime $ by setting,
$\Phi ^{(r)}(j_x^rs)=j_{\phi (x)}^r
(\Phi \circ s\circ \phi ^{-1})$,
$\forall j_x^rs\in J^rE$.

An automorphism of $p$ is a pair of diffeomorphisms
$\Phi \colon E\to E$, $\phi \colon M\to M$ such that
$p\circ \Phi =\phi \circ p$.
The set of all automorphisms of $p$ is a group denoted
by $\mathrm{Aut}E$.

For every---not necessarily $p$-projectable---vector
field $Y\in \mathfrak{X}(E)$ a unique vector field
$Y^{(1)}\in \mathfrak{X}(J^1E)$ exists (called
the $1$-jet prolongation of $Y$) such that, 1st)
$Y^{(1)}$ is $p_{10}$-projectable onto $Y$, and
2nd) $L_{Y^{(1)}}\mathcal{C}\subseteq \mathcal{C}$,
where $\mathcal{C}$ is the contact system; i.e.,
the differential system generated by the contact
$1$-forms $\theta ^\alpha =dy^\alpha -y_i^\alpha dx^i$
in $\Omega ^1(J^1E)$.

If $Y$ is a $p$-vertical vector field (the only case
that we consider below), then the formulas of $1$-jet
prolongation are as follows (cf.\ \cite[Section 2.4]{MP}):

\begin{align}
Y^{(1)}
& =v^\alpha
\frac{\partial }{\partial y^\alpha }
+v_i^\alpha
\frac{\partial }{\partial y_i^\alpha },
\quad
v_i^\alpha
=\frac{\partial v^\alpha }{\partial x^i}
+\frac{\partial v^\alpha }
{\partial y^\beta }y_i^\beta ,
\label{i.c.t.} \\
Y & =v^\alpha
\frac{\partial }{\partial y^\alpha },
\quad v^\alpha \in C^\infty (E).
\nonumber
\end{align}

\begin{lemma}
[\cite{MP}]
\label{l1}
For every $p$-vertical vector field
$Y\in \mathfrak{X}^v(E)$ the value of the vector
field $Y^{(1)}$ at a point $j_x^1s\in J^1E$
depends only on $j_x^1\bigl( Y\circ s\bigr) $.
\end{lemma}

\subsection{Bundles of connections
\label{BundleOfConnections}}

An automorphism of a principal $G$-bundle
$\pi \colon P\to M$ is an equivariant diffeomorphism
$\Phi \colon P\to  P$; i.e., $\Phi (u\cdot g)
=\Phi (u)\cdot g$, $\forall u\in P$, $\forall g\in G$.
We denote by $\mathrm{Aut}P$ the group of all
automorphisms of $P$ under composition. Every
$\Phi \in \mathrm{Aut}P$ determines a unique
diffeomorphism $\phi \colon M\to M$, such that
$\pi \circ \Phi =\phi \circ \pi $. If $\phi $
is the identity map, then $\Phi $ is said to be
a gauge transformation (cf.\ \cite[3.2.1]{Bl}).
We denote by $\mathrm{Gau}P\subset \mathrm{Aut}P$
the subgroup of all gauge transformations.

A vector field $X\in \mathfrak{X}(P)$
is said to be $G$-invariant if
$R_g\cdot X=X$, $\forall g\in G$.
If $\Phi _t$ is the flow of a vector
field $X\in \mathfrak{X}(P)$, then $X$
is $G$-invariant if and only if
$\Phi _t\in \mathrm{Aut}P$,
$\forall t\in \mathbb{R}$.
Because of this we denote the Lie
subalgebra of $G$-invariant vector fields
on $P$ by $\mathrm{aut}P
\subset \mathfrak{X}(P)$. Each $G$-invariant
vector field on $P$ is $\pi$-projectable.
Similarly, a $\pi $-vertical vector field
$X\in \mathfrak{X}( P)$ is $G$-invariant
if and only if $\Phi _t\in \mathrm{Gau}P$,
$\forall t\in \mathbb{R}$. We denote by
$\mathrm{gau}P\subset \mathrm{aut}P$
the ideal of all $\pi$-vertical
$G$-invariant vector fields on $P$,
which is called the gauge algebra of $P$.

The group $G$ acts on $TP$ by setting
$X\cdot g=(R_g)_\ast (X)$,
$\forall X\in TP$, $\forall g\in G$.
The global sections of the quotient vector
bundle $T(P)/G$ can naturally be identified
to $\mathrm{aut}P$; i.e., $\mathrm{aut}P
\cong \Gamma (M,T(P)/G)$. Similarly,
$\mathrm{gau}P\cong \Gamma (M,\mathrm{ad}P)$,
where $\mathrm{ad}P$ denotes the adjoint
bundle: The bundle associated to $P$
by the adjoint representation of $G$
on its Lie algebra $\mathfrak{g}$,
denoted by $\pi _{\mathfrak{g}}
\colon \mathrm{ad}P\to M$; that is,
$\mathrm{ad}P=(P\times \mathfrak{g})/G$,
where the action of $G$ on
$P\times\mathfrak{g}$ is defined by
\[
\left(
u,B
\right)
\cdot g
=\left(
u\cdot g,\mathrm{Ad}_{g^{-1}}
(B)
\right) ,
\quad
\forall u\in P,
\; \forall B\in \mathfrak{g},
\; \forall g\in G.
\]
The $G$-orbit in $\mathrm{ad}P$ of a pair
$(u,B)\in P\times\mathfrak{g}$ is denoted
by $(u,B)_{\mathrm{ad}}$.

An exact sequence of vector bundles over
$M$ (the so-called Atiyah sequence) holds,
$0\to \mathrm{ad}P\to  T(P)/G
\overset{\pi _\ast}{\longrightarrow }TM
\to 0$. The fibres $(\mathrm{ad}P)_x$
are endowed with a Lie-algebra structure
determined by
$[(u,B)_{\mathrm{ad}},(u,B^\prime )_{\mathrm{ad}}]
=(u,-[B,B^\prime ])_{\mathrm{ad}}$,
for all $u\in \pi ^{-1}(x)$,
$B,B^\prime \in \mathfrak{g}$, where
$[\cdot ,\cdot ]$ denotes the bracket
in $\mathfrak{g}$. The sign of the bracket
above is needed in order to ensure
that the natural identification
$\mathrm{gau}P\cong \Gamma (M,\mathrm{ad}P)$
is a Lie-algebra isomorphism, when
$\mathrm{gau}P$ is considered as a Lie
subalgebra of $\mathfrak{X}(P)$.

Let $X^{h_\Gamma }\in \mathfrak{X}(P)$
be the horizontal lift of a vector field
$X\in \mathfrak{X}(M)$ with respect
to a connection $\Gamma $ on
$\pi \colon P\to M$. The vector field
$X^{h_\Gamma }$ is $G$-invariant
and projects onto $X$ (cf.\
\cite[II. Proposition\ 1.2]{KN}).
Hence we have a splitting of the Atiyah
sequence, $s_\Gamma \colon TM\to T(P)/G$,
$s_\Gamma (X)=X^{h_\Gamma }$. Conversely,
any splitting $\sigma \colon TM\to T(P)/G$
of the Atiyah sequence (i.e., $s$ is a vector
bundle homomorphism such that
$\pi _\ast\circ s=\mathrm{id}_{TM}$)
comes from a unique connection on $P$
so that there is a natural bijection
between connections on $P$ and splittings
of the Atiyah sequence. We thus define
the bundle of connections $p\colon C=C(P)\to M$
as the sub-bundle of $\mathrm{Hom}(TM,T(P)/G)$
determined by all $\mathbb{R}$-linear mappings
$\lambda\colon T_xM\to (T(P)/G)_x$ such that
$\pi _\ast \circ \lambda =\mathrm{id}_{T_xM}$.
Connections on $P$ can be identified
to the global sections of $p\colon C\to M$.
We also denote by $s_\Gamma \colon M\to C$
the section of the bundle of connections induced
by $\Gamma $.

An element $\lambda \colon T_xM\to (T(P)/G)_x$
of the bundle $C$ over a point $x\in M$ is nothing
but a `connection at a point $x$'; i.e., $\Lambda $
induces a complementary subspace $H_u$
of the vertical subspace $V_u(P)\subset T_u(P)$
for every $u\in \pi ^{-1}(x)$. Any other connection
at $x$ can be written as $\lambda ^\prime
=h+\lambda $, where $h\colon T_xM
\to (\mathrm{ad}P)_x$ is a linear map. Hence $C$
is an affine bundle modelled over the vector bundle
$\mathrm{Hom}(TM,\mathrm{ad}P)\cong T^\ast M
\otimes \mathrm{ad}P$.

Let $(U;x^i)$ be a coordinate open domain in $M$
such that $\pi ^{-1}(U)\cong U\times G$. For every
$B\in \mathfrak{g}$ we define a flow of gauge
transformations over $U$ by setting
$\varphi _t^B(x,g)=(x,\exp (tB)\cdot g)$, $x\in U$.
Let $\tilde{B}$ be the corresponding infinitesimal
generator. If $(B_1,\dotsc,B_m)$ is a basis
of $\mathfrak{g}$, then
$\tilde{B}_1,\dotsc,\tilde{B}_m$ is a basis
of $\Gamma (U,\mathrm{ad}P)$.
Let $p\colon C\to M$ be the bundle of connections
of $P$. The horizontal lift with respect to $\Gamma$
of the basic vector field $\partial /\partial x^i$
is given as follows:
\begin{equation*}
s_\Gamma
\Bigl(
\frac{\partial }{\partial x^i}
\Bigr)
=
\Bigl(
\frac{\partial }{\partial x^i}
\Bigr) ^{h_\Gamma }
=\frac{\partial }{\partial x^i}
-\left(
A_i^\alpha \circ s_\Gamma
\right)
\tilde{B}_\alpha .
\end{equation*}
The functions $(x^i,A_j^\alpha )$, $i,j=1,\dotsc,n=\dim M$,
$1\leq \alpha\leq m=\dim G$, induce a coordinate system
on $p^{-1}(U)=C(\pi ^{-1}U)$ (cf.\ \cite{CM}).

Each automorphism $\Phi \in \mathrm{Aut}P$ acts
on connections of $P$ by pulling back connection forms;
that is, $\Gamma ^\prime =\Phi (\Gamma )$ where
$\omega _{\Gamma ^\prime }=(\Phi ^{-1})^\ast \omega _\Gamma $
(cf.\ \cite[II. Proposition\ 6.2-(b)]{KN}).
If $\Psi\in \mathrm{Aut}P$ is another automorphism,
then $(\Psi \circ \Phi )(\Gamma )$
$=\Psi (\Phi (\Gamma ))$. For each $\Phi \in \mathrm{Aut}P$
there exists a unique diffeomorphism $\Phi _C\colon C\to C$
such that $p\circ \Phi _{\mathcal{C}}=\Phi \circ p$, where
$\Phi \colon M\to M$ is the diffeomorphism induced by $\Phi $
on the base manifold. We thus obtain a group homomorphism
$\mathrm{Aut}P\to \mathrm{Diff}C$. For every connection
$\Gamma $ on $P$ we have
$\Phi _C\circ s_\Gamma =s_{\Phi (\Gamma )}$.

If $\Phi _t$ is the flow of a $G$-invariant vector field
$X\in \mathrm{aut}P$, then $(\Phi _t)_C$ is a one-parameter
group in $\mathrm{Diff}C$ with infinitesimal generator
denoted by $X_C$, and the map $\mathrm{aut}P\to \mathfrak{X}(C)$,
$X\mapsto X_C$ is a Lie-algebra homomorphism.

\subsection{Affine-bundle structures\label{AffineBundles}}

Let $V(p)=\{ X\in TE:p_\ast X=0\} $ be the vertical subbundle
of a fibred manifold $p\colon E\to M$.

\begin{enumerate}
\item[(a)]
Let $p\colon E\to M$ be an affine bundle modelled
over the vector bundle $p_W\colon W\to M$.
The directional derivative determines an isomorphism
of vector bundles over $E$, $p^\ast W\cong V(p)$,
$(e,w)\mapsto X_{e,w}$, $p(e)=p_W(w)=x$, where $X_{e,w}$
is the tangent vector at $t=0$ to the curve $t\mapsto tw+e$,
which takes values in the fibre $p^{-1}(x)$. In coordinates,
$X_{e,w}=w^\alpha (w)(\partial /\partial e^\alpha )_e$.

\item[(b)]
Moreover, if $p\colon E\to M$ is an arbitrary surjective
submersion, then the projection $p_{10}\colon J^1E\to E$
is endowed with an affine-bundle structure modelled
over $p^\ast T^\ast M\otimes V(p)$. In fact, every jet
$j_x^1s\in (p_{10})^{-1}(e)$, with $s(x)=e$, can be
identified to the section $s_{\ast,e}\colon T_xM\to T_eE$
of $p_{\ast ,e}\colon T_eE\to T_xM$. Hence, if
$j_x^1s^\prime \in (p_{10})^{-1}(e)$ is another jet,
then $p_{\ast ,e}\circ(s_{\ast ,e}^\prime -s_{\ast ,e})=0$
and accordingly $s_{\ast ,e}^\prime -s_{\ast ,e}$ takes
values into $V_e(p)$. Therefore,
$s_{\ast ,e}^\prime -s_{\ast ,e}$ determines an element
in $\mathrm{Hom}(T_xM\otimes V_e(p))=T_x^\ast M
\otimes V_e(p)$. From (a) it follows an isomorphism,
$p_{10}^\ast (p^\ast T^\ast M\otimes V(p))
=T^\ast M\otimes _{J^1C}V(p)\cong V(p_{10})$.
\end{enumerate}

\subsection{The Hessian metric\label{HessianMetric}}

If $p\colon E\to M$ be an affine bundle modelled
over the vector bundle $p_W\colon W\to M$, then,
according to the item (a) in Section \ref{AffineBundles},
every $w\in (p_W)^{-1}(x)$ induces a vector field
along the fibre $X_w\in \mathfrak{X}(p^{-1}(x))$,
$X_w(e)=X_{e,w}$, $\forall e\in p^{-1}(x)$. For every
$f\in C^\infty (E)$ and every $e\in E$, with $x=p(e)$,
a bilinear form
\[
\mathrm{Hess}_e(f)\colon V_e(p)\times V_e(p)\to \mathbb{R}
\]
can be defined as follows: $\mathrm{Hess}_e(f)(w_1,w_2)
=X_{w_2}(e)((d_{E/M}f)X_{w_1})$, where the canonical
isomorphism $W_x\cong V_e(p)$, defined in the item (a)
in Section \ref{AffineBundles}, has be used and $d_{E/M}$
denotes the fibred derivative, e.g., see \cite{GS}.
As this form is proved to be symmetric,
$e\mapsto \mathrm{Hess}_e(f)$ defines a section
of the vector bundle $S^2V^\ast (p)\cong p^\ast S^2W^\ast $.

\section{Hamilton-Cartan equations\label{sectionHC}}

Let $p\colon E\to M$ be a fibred manifold, $\dim M=n$,
$\dim E=m+n$, where $M$ is assumed to be connected
and oriented by a volume form $\mathbf{v}$. Below,
Latin indices run from $1$ to $n$, and Greek indices
run from $1$ to $m$. The solutions to the Hamilton-Cartan
equations for a density $\Lambda =L\mathbf{v}$,
$L\in C^\infty (J^1E)$ on $p$, are the sections
$\bar{s}\colon M\to J^1E$ of the canonical projection
$p_1\colon J^1E\to M$ such that,
\begin{equation}
\bar{s}^\ast
\left(
i_Xd\Theta _\Lambda
\right)
=0,
\quad
\forall X\in \mathfrak{X}^v(J^1E), \label{HC}
\end{equation}
where

\begin{description}
\item[$\mathrm{(i)}$]
$\Theta _\Lambda
=(-1)^{i-1}(\partial L/\partial y_i^\alpha )
\theta ^\alpha \wedge \mathbf{v}_i
+L\mathbf{v}$ is the Poincar\'e-Cartan
form attached to $\Lambda $ (cf.\ \cite{GS},
\cite{MP}),

\item[$\mathrm{(ii)}$]
$\mathfrak{X}^v(J^1E)$ denotes the Lie algebra
of $p_1$-vertical vector fields,

\item[$\mathrm{(iii)}$]
$\theta^\alpha =dy^\alpha -y_i^\alpha dx^i$ are
the standard contact forms on the $1$-jet bundle,
and

\item[$\mathrm{(iv)}$]
$(x^i,y^\alpha ,y_i^\alpha )$ is the coordinate
system on $J^1E$ induced by a fibred coordinate
system $(x^i,y^\alpha )$ for the submersion $p$,
adapted to the given volume form; i.e.,
$\mathbf{v}=dx^1\wedge\ldots \wedge dx^n$ and
$\mathbf{v}_i
=(-1)^{i-1}i_{\partial /\partial x^i}\mathbf{v}$.
\end{description}

\begin{lemma}
\label{L1}
A section $\bar{s}\colon M\to J^1E$
of $p_1\colon J^1E\to M$ is a solution
to the H-C equations \emph{\eqref{HC}}
if and only if the following equations hold:
\begin{equation*}
\bigl(
s_j^\beta -\bar{s}_j^\beta
\bigr)
\Bigl(
\dfrac{\partial ^2L}
{\partial y_i^\alpha \partial y_j^\beta }
\circ \bar{s}
\Bigr)
=0,
\quad
1\leq i\leq n,1\leq \alpha \leq m,
\end{equation*}
\begin{equation*}
-\dfrac{\partial }{\partial x^j}
\Bigl(
\dfrac{\partial L}{\partial y_j^\alpha }
\circ\bar{s}
\Bigr)
+\dfrac{\partial L}{\partial y^\alpha }
\circ\bar{s}
+\bigl(
s_j^\beta -\bar{s}_j^\beta
\bigr)
\Bigl(
\dfrac{\partial ^2L}
{\partial y^\alpha \partial y_j^\beta }
\circ\bar{s}
\Bigr)
=0,
\end{equation*}
for $1\leq \alpha \leq m$, where
$s^\alpha =y^\alpha \circ\bar{s}$,
$s_i^\alpha =\partial s^\alpha /\partial x^i$,
and $\bar{s}_i^\alpha =y_i^\alpha \circ \bar{s}$.
\end{lemma}

\begin{proof}
As a simple computation shows, we have
\begin{equation}
d\Theta _\Lambda
=\theta ^\beta \wedge
\left(
(-1)^jd
\Bigl(
\frac{\partial L}{\partial y_j^\beta }
\Bigr)
\wedge \mathbf{v}_j
+\frac{\partial L}{\partial y^\beta }
\mathbf{v}
\right) . \label{d_Theta}
\end{equation}
Hence
\begin{align*}
\bar{s}^\ast
\left(
i_{\partial /\partial y_i^\alpha }
d\Theta _\Lambda
\right)
& =(-1)^{j-1}
\Bigl(
\dfrac{\partial ^2L}
{\partial y_i^\alpha \partial y_j^\beta }
\circ \bar{s}
\Bigr)
\bar{s}^\ast
\theta^\beta \wedge\mathbf{v}_j, \\
\bar{s}^\ast
\left(
i_{\partial /\partial y^\alpha }
d\Theta _\Lambda
\right)
& =\left(
-\dfrac{\partial }{\partial x^j}
\Bigl(
\dfrac{\partial L}{\partial y_j^\alpha }
\circ \bar{s}
\Bigr)
+\dfrac{\partial L}{\partial y^\alpha }
\circ \bar{s}
\right)
\mathbf{v} \\
&
\quad
\qquad
\mathbf{+}(-1)^{j-1}
\Bigl(
\dfrac{\partial ^2L}
{\partial y^\alpha \partial y_j^\beta }
\circ \bar{s}
\Bigr)
\bar{s}^\ast \theta ^\beta \wedge \mathbf{v}_j,
\end{align*}
and the formulas in the statement follow.
\end{proof}

If $\bar{s}=j^1s$ is a holonomic section,
then $\bar{s}$ is a solution to the H-C
equations if and only if $s$ is a solution
to the Euler-Lagrange equations. If $L$ is
regular, then the converse holds true:
Every solution to the H-C equations, is
of the form $\bar{s}=j^1s$, $s$ being
a solution to the E-L equations. Hence,
for regular variational problems,
H-C equations are equivalent to E-L equations;
but this is no longer true for non-regular
densities, as is the case for the Yang-Mills
Lagrangian.

\section{Jacobi fields\label{sectionJacobi}}

\subsection{Jacobi fields introduced}

Let $p\colon E\to M$ be a fibred manifold
and let $\Omega ^1(E/M)=\Gamma (M,V^\ast (p))$.
Let
\[
\mathcal{E}^\Lambda \colon \Gamma (p)
\to \Omega ^1(E/M)\otimes _{C^\infty (M)}\Omega ^n(M)
\]
be the Euler-Lagrange operator
of $\Lambda =L\mathbf{v}$, $L\in C^\infty (J^1E)$,
which is the second-order differential operator
locally given on a fibred coordinate system
$(U;x^i,y^\alpha )$ for the submersion $p$
and for every section $s$ of $p|_U$ by,
\[
\mathcal{E}^\Lambda (s)
=\left(
\mathcal{E}_\alpha ^\Lambda \circ j^2s
\right)
d_{E/M}y^\alpha \otimes \mathbf{v},
\]
where the functions
$\mathcal{E}_\alpha ^\Lambda \in C^\infty (J^2E)$
are defined as follows:
\[
\mathcal{E}_\alpha ^\Lambda
\left(
j_x^2s
\right)
=\frac{\partial L}{\partial y^\alpha }(j_x^1s)
-\frac{\partial }{\partial x^j}
\left(
\frac{\partial L}{\partial y_j^\alpha }
\circ j^1s
\right) (x).
\]

The linearisation of $\mathcal{E}^\Lambda $
at $s\in \Gamma (p)$ is the operator
\[
\mathcal{L}_s\mathcal{E}^\Lambda
\colon T_s\Gamma (p)
=\Gamma
\left(
M,s^\ast V(p)
\right)
\to \Omega ^1(E/M)\otimes _{C^\infty (M)}
\Omega ^n(M)
\]
defined as follows.

If $S\colon(-\varepsilon ,\varepsilon )
\times U\to E$ is a one-parameter family
of sections, i.e.,
$p\circ S_t=\mathrm{id}_U$, $|t|<\varepsilon $,
then a vector field $X\in \Gamma (U,s^\ast V(p))$
along $s=S_0$---called the `initial velocity'
of $S$---is defined to be the tangent vector
$X(x)\in V_{s(x)}(p)$ at $t=0$ to the curve
$t\mapsto S_t(x)$, which takes values in the fibre
$p^{-1}(x)$ for every point $x\in U$.
Expanding $y^\alpha \circ S_t$ up to second order,
we obtain $y^\alpha \circ S_t
=y^\alpha \circ s+tv^\alpha +t^2f^\alpha $,
for certain functions $v^\alpha \in C^\infty (U)$,
$f^\alpha
\in C^\infty ((-\varepsilon ,\varepsilon )
\times U)$.
Hence,
\[
X(x)=v^\alpha (x)
\left(
\partial /\partial y^\alpha
\right) _{s(x)},
\quad
\forall x\in U.
\]
Therefore, every vector field
$X\in \Gamma (U,s^\ast V(p))$
is the initial velocity of a one-parameter family
of sections $S$, and we define,
\[
\mathcal{L}_s\mathcal{E}^\Lambda(X)
=\left.
\frac{\partial }{\partial t}\right\vert _{t=0}
\mathcal{E}^\Lambda
\left(
S_t
\right) .
\]
The definition makes sense as it does not depend
on the particular one-parameter family of sections
chosen. In fact,
\begin{align*}
\left.
\frac{\partial }{\partial t}\right\vert _{t=0}
\left(
\mathcal{E}_\alpha ^\Lambda \circ j^2S_t
\right)
& =\left(
\frac{\partial \mathcal{E}_\alpha ^\Lambda }
{\partial y^\beta }\circ j^2s
\right)
v^\beta
+\left(
\frac{\partial\mathcal{E}_\alpha ^\Lambda }
{\partial y_j^\beta }\circ j^2s
\right)
\frac{\partial v^\beta }{\partial x^j} \\
& +\left(
\frac{\partial \mathcal{E}_\alpha ^\Lambda }
{\partial y_{jk}^\beta }\circ j^2s
\right)
\frac{\partial ^2v^\beta }
{\partial x^j\partial x^k}.
\end{align*}
This expression also shows that
$\mathcal{L}_s\mathcal{E}^\Lambda (X)$
depends linearly on $X$.

\begin{definition}
A vector field $X\in \Gamma (M,s^\ast V(p))$
defined along an extremal $s$ of a Lagrangian
density $\Lambda $ on $p\colon E\to M$ is said
to be a \emph{Jacobi field} if
$X\in \ker \mathcal{L}_s\mathcal{E}^\Lambda $.
\end{definition}

\begin{proposition}
A vector field $X\in \Gamma (M,s^\ast V(p))$
defined along an extremal $s$ of a Lagrangian
density $\Lambda $ on $p\colon E\to M$
is a Jacobi field if and only if
the following equation holds:
\begin{equation}
(j^1s)^\ast
\left(
i_YL_{\tilde{X}^{(1)}}d\Theta _\Lambda
\right)
=0,\quad \forall Y\in \mathfrak{X}^v(J^1E),
\label{Jacobi}
\end{equation}
where $\tilde{X}\in \mathfrak{X}^v(E)$
is an arbitrary $p$-vertical extension of $X$.
\end{proposition}

\begin{proof}
The equation \eqref{Jacobi} does not depend
on the vertical extension chosen. In fact,
\begin{equation}
(j^1s)^\ast
\left(
i_YL_{\tilde{X}^{(1)}}d\Theta _\Lambda
\right)
=(j^1s)^\ast
\left(
L_{\tilde{X}^{(1)}}
\left(
i_Yd\Theta _\Lambda
\right)
\right)
+(j^1s)^\ast
\bigl(
i_{[Y,\tilde{X}^{(1)}]}
d\Theta _\Lambda \bigr) ,\label{formula}
\end{equation}
and the second term on the right-hand side
vanishes, as $s$ is an extremal. Hence
\[
(j^1s)^\ast
\left(
i_YL_{\tilde{X}^{(1)}}d\Theta _\Lambda
\right)
=d\left(
(j^1s)^\ast
\left(
i_{\tilde{X}^{(1)}}i_Yd\Theta _\Lambda
\right)
\right)
+(j^1s)^\ast
\left(
i_{\tilde{X}^{(1)}}
d\left(
i_Y\Theta _\Lambda
\right)
\right) ,
\]
and we conclude by simply applying Lemma \ref{l1}.

Moreover, from the formula \eqref{d_Theta}
the following identity is obtained:
\[
d\Theta _\Lambda
=\mathcal{E}_\alpha ^\Lambda \theta ^\alpha
\wedge \mathbf{v}
+(-1)^i\frac{\partial ^2L}
{\partial y_i^\alpha \partial y^\beta }
\theta ^\alpha \wedge \theta ^\beta
\wedge \mathbf{v}_i
+(-1)^i\frac{\partial ^2L}
{\partial y_i^\alpha \partial y_j^\beta }
\theta ^\alpha \wedge \theta _j^\beta
\wedge \mathbf{v}_i,
\]
where $\theta _j^\beta
=dy_j^\beta -y_{(jk)}^\beta dx^k$.
If $\tilde{\Phi }_t$ is the flow of $\tilde{X}$,
then
$j^1S_t=J^1(\tilde{\Phi }_t)\circ j^1s$
and from the previous formula we have
\begin{align*}
\left(
j^1s
\right) ^\ast
\left\{
J^1(\tilde{\Phi }_t)^\ast
\left(
i_Yd\Theta _\Lambda
\right)
\right\}
& =\left(
j^1S_t
\right)
^\ast
\left(
i_Yd\Theta _\Lambda
\right) \\
& =\left(
\mathcal{E}_\alpha ^\Lambda \circ j^2S_t
\right)
\theta ^\alpha (Y)\mathbf{v}.
\end{align*}
Taking derivatives with respect to $t$
at $t=0$ in this formula, we have
\[
(j^1s)^\ast
\left(
L_{\tilde{X}^{(1)}}
\left(
i_Yd\Theta _\Lambda
\right)
\right)
=\mathcal{L}_s\mathcal{E}^\Lambda (X)
\mathbf{v},
\]
and we can conclude by simply applying
the formula \eqref{formula} recalling
that $s$ is an extremal.
\end{proof}

Let $\mathcal{S}_\Lambda $ (resp.\
$\mathcal{\bar{S}}_\Lambda $) denote the set
of solutions to the E-L equations
(resp.\ H-C equations) attached
to a Lagrangian density $\Lambda=L\mathbf{v}$,
$L\in C^\infty (J^1E)$.

\begin{remark}
If $S_t\in \mathcal{S}_\Lambda $ is a one-parameter
family of extremals, then its initial velocity $X$
is readily seen to be a Jacobi field along
the extremal $s=S_0$; in this case, $X$ is said
to be `integrable' (e.g., see
\cite[Definition 1.2]{LeW}, \cite[Section 2.6]{LiW}).
Although important examples of non-integrable Jacobi
fields exist, usually Jacobi fields
along $s$ are considered as the tangent space
$T_s\mathcal{S}_\Lambda $ at an extremal
$s\in \mathcal{S}_\Lambda $ to the `manifold'
of solutions to the E-L equations for $\Lambda $.
By the same token, we give the following
\end{remark}

\begin{definition}
The tangent space
$T_{\bar{s}}\mathcal{\bar{S}}_\Lambda $
at a solution
$\bar{s}\in \mathcal{\bar{S}}_\Lambda $
to the `manifold' of solutions to the H-C equations
for $\Lambda $ is defined to be the space of vector
fields $\bar{X}\in \Gamma (M,\bar{s}^\ast V(p_1))$
that satisfy the equation
\[
\bar{s}^\ast
\left(
i_YL_{\tilde{X}}d\Theta _\Lambda
\right)
=0,
\quad
\forall Y\in \mathfrak{X}^v(J^1E),
\]
where $\tilde{X}\in \mathfrak{X}^v(J^1E)$
is a $p_1$-vertical extension of $\bar{X}$.
\end{definition}

\subsection{The embedding
$T_s\mathcal{S}_\Lambda \hookrightarrow
T_{j^1s}\mathcal{\bar{S}}_\Lambda $}

\begin{proposition}
\label{p2}
For every $s\in \mathcal{S}_\Lambda $,
there is an embedding
\begin{align*}
T_s\mathcal{S}_\Lambda
& \hookrightarrow
T_{j^1s}\mathcal{\bar{S}}_\Lambda , \\
X & \mapsto \tilde{X}^{(1)}\circ j^1s,
\end{align*}
where $\tilde{X}$ is any $p$-vertical extension
of $X\in \Gamma (M,s^{\ast }V(p))$ to $E$.
If $\Lambda $ is regular, then
$T_s\mathcal{S}_\Lambda
\cong T_{j^1s}\mathcal{\bar{S}}_\Lambda $,
$\forall s\in \mathcal{S}_\Lambda $.
\end{proposition}

\begin{proof}
As a straightforward---but rather
long---computation shows, a vector field
$\bar{X}\in \Gamma (M,\bar{s}^\ast V(p_1))$
with local expression
\[
\bar{X}=v^\alpha
\left.
\frac{\partial }{\partial y^\alpha }
\right| _{\bar{s}}
+v_i^\alpha
\left.
\frac{\partial }{\partial y_i^\alpha }
\right| _{\bar{s}},
\quad
v^\alpha ,v_i^\alpha \in C^\infty (M),
\]
belongs to $T_{\bar{s}}\mathcal{\bar{S}}_\Lambda $
if and only if the following two equations hold:

\begin{equation}
\begin{array}{cl}
0= &
\left(
\dfrac{\partial ^2L}
{\partial y^\alpha \partial y^\sigma }
\circ \bar{s}
-\dfrac{\partial ^3L}
{\partial x^i\partial y^\alpha
\partial y_i^\sigma }
\circ \bar{s}
+\Bigl(
\dfrac{\partial s^\beta }
{\partial x^i}-\bar{s}_i^\beta
\Bigr)
\Bigl(
\dfrac{\partial ^3L}
{\partial y^\alpha \partial y^\sigma
\partial y_i^\beta }
\circ \bar{s}
\Bigr)
\right. \\
& \left.
-\dfrac{\partial s^\gamma }{\partial x^j}
\Bigl(
\dfrac{\partial ^3L}
{\partial y^\alpha \partial y^\gamma
\partial y_j^\sigma }
\circ \bar{s}
\Bigr)
-\dfrac{\partial \bar{s}_h^\gamma }
{\partial x^j}
\Bigl(
\dfrac{\partial ^3L}
{\partial y^\alpha \partial y_h^\gamma
\partial y_j^\sigma }
\circ \bar{s}
\Bigr)
\right)
v^\alpha \\
& +\Bigl(
\dfrac{\partial ^2L}
{\partial y^\sigma
\partial y_i^\alpha }
\circ \bar{s}
-\dfrac{\partial ^2L}
{\partial y^\alpha
\partial y_i^\sigma }
\circ \bar{s}
\Bigr)
\dfrac{\partial v^\alpha }
{\partial x^i}
-\Bigl(
\dfrac{\partial ^2L}
{\partial y_j^\alpha
\partial y_i^\sigma }
\circ \bar{s}
\Bigr)
\dfrac{\partial v_j^\alpha }
{\partial x^i} \\
& +\left(
\Bigl(
\dfrac{\partial s^\beta }
{\partial x^j}
-\bar{s}_j^\beta
\Bigr)
\Bigl(
\dfrac{\partial ^3L}
{\partial y^\sigma
\partial y_i^\alpha
\partial y_j^\beta }
\circ \bar{s}
\Bigr)
-\dfrac{\partial ^3L}
{\partial x^j\partial y_i^\alpha
\partial y_j^\sigma }
\circ \bar{s}
\right. \\
& \left.
-\dfrac{\partial s^\beta }{\partial x^j}
\Bigl(
\dfrac{\partial ^3L}
{\partial y^\beta \partial y_i^\alpha
\partial y_j^\sigma }
\circ \bar{s}
\Bigr)
-\dfrac{\partial \bar{s}_h^\beta }
{\partial x^j}
\Bigl(
\dfrac{\partial ^3L}
{\partial y_h^\beta \partial y_i^\alpha
\partial y_j^\sigma }
\circ \bar{s}
\Bigr)
\right)
v_i^\alpha ,
\end{array}
\label{Jacobi_I}
\end{equation}
\begin{equation}
\begin{array}
[c]{cl}
0= &
\Bigl(
\dfrac{\partial ^2L}
{\partial y_i^\alpha \partial y_j^\sigma }
\circ \bar{s}
\Bigr)
\Bigl(
\dfrac{\partial v^\alpha }
{\partial x^i}
-v_i^\alpha
\Bigr)
+\Bigl(
\dfrac{\partial s^\alpha }
{\partial x^i}
-\bar{s}_i^\alpha
\Bigr)
\Bigl(
\dfrac{\partial ^3L}
{\partial y^\beta \partial y_i^\alpha
\partial y_j^\sigma }
\circ\bar{s}
\Bigr)
v^\beta
\medskip \\
& +\Bigl(
\dfrac{\partial s^\alpha }
{\partial x^i}
-\bar{s}_i^\alpha
\Bigr)
\Bigl(
\dfrac{\partial ^3L}
{\partial y_i^\alpha
\partial y_k^\beta
\partial y_j^\sigma }
\circ\bar{s}
\Bigr)
v_k^\beta ,
\end{array}
\label{Jacobi_II}
\end{equation}
where $s=p_{10}\circ \bar{s}$,
$s^\alpha =y^\alpha \circ\bar{s}
=y^\alpha \circ s$, and
$\bar{s}_i^\alpha
=y_i^\alpha \circ \bar{s}$.
Along a holonomic section
$\bar{s}=j^1s$, the equations
\eqref{Jacobi_I} and \eqref{Jacobi_II}
become respectively,

\begin{equation}
\begin{array}{ll}
0= &
\left(
\dfrac{\partial ^2L}
{\partial y^\alpha \partial y^\sigma }
\circ j^1s
-\dfrac{\partial ^3L}
{\partial x^i\partial y^\alpha
\partial y_i^\sigma }
\circ j^1s
-\dfrac{\partial s^\gamma }{\partial x^j}
\Bigl(
\dfrac{\partial ^3L}
{\partial y^\alpha
\partial y^\gamma
\partial y_j^\sigma }
\circ j^1s
\Bigr)
\right. \\
& \left.
-\dfrac{\partial ^2s^\gamma }
{\partial x^h\partial x^j}
\Bigl(
\dfrac{\partial ^3L}
{\partial y^\alpha
\partial y_h^\gamma
\partial y_j^\sigma }
\circ j^1s
\Bigr)
\right)
v^\alpha \\
& +\Bigl(
\dfrac{\partial ^2L}
{\partial y^\sigma
\partial y_i^\alpha }
\circ j^1s
-\dfrac{\partial ^2L}
{\partial y^\alpha
\partial y_i^\sigma }
\circ j^1s
\Bigr)
\dfrac{\partial v^\alpha }
{\partial x^i}
-\Bigl(
\dfrac{\partial ^2L}
{\partial y_j^\alpha
\partial y_i^\sigma }
\circ j^1s
\Bigr)
\dfrac{\partial v_j^\alpha }
{\partial x^i} \\
& -\left(
\dfrac{\partial ^3L}
{\partial x^j\partial y_i^\alpha
\partial y_j^\sigma }
\circ j^1s
+\dfrac{\partial s^\beta }
{\partial x^j}
\Bigl(
\dfrac{\partial ^3L}
{\partial y^\beta
\partial y_i^\alpha
\partial y_j^\sigma }
\circ j^1s
\Bigr)
\right. \\
& \left.
+\dfrac{\partial ^2s^\beta }
{\partial x^h\partial x^j}
\Bigl(
\dfrac{\partial ^3L}
{\partial y_i^\alpha
\partial y_h^\beta
\partial y_j^\sigma }
\circ j^1s
\Bigr)
\right)
v_i^\alpha ,
\end{array}
\label{Jacobi_I_hol}
\end{equation}
\begin{equation}
0=\Bigl(
\dfrac{\partial ^2L}
{\partial y_i^\alpha \partial y_j^\sigma }
\circ j^1s
\Bigr)
\Bigl(
\dfrac{\partial v^\alpha }{\partial x^i}
-v_i^\alpha
\Bigr) .
\label{Jacobi_II_hol}
\end{equation}

In addition, if $L$ is regular,
then the equation \eqref{Jacobi_II_hol}
is equivalent to saying that
$v_i^\alpha
=\partial v^\alpha /\partial x^i$,
and from the formula \eqref{i.c.t.}
and Lemma \ref{l1} we conclude that
$\bar{X}$ is the $1$-jet prolongation
of a Jacobi field, i.e., $\bar{X}=X^{(1)}$.
\end{proof}

\begin{remark}
\label{remm}
From the formula \eqref{Jacobi_I_hol} we deduce that
a vector field
\begin{align*}
X  & \in \Gamma (M,s^\ast V(p)), \\
X  & =v^\alpha
\left.
\frac{\partial }{\partial y^\alpha }
\right| _s,
\quad
v^\alpha \in C^\infty (M),
\end{align*}
belongs to $T_s\mathcal{S}_\Lambda $
if and only if the following equations
hold:

\begin{equation*}
\begin{array}{ll}
0= & \!\!
\left(
\dfrac{\partial ^2L}
{\partial y^\alpha \partial y^\sigma }
\circ j^1s
-\dfrac{\partial ^3L}
{\partial x^i\partial y^\alpha
\partial y_i^\sigma }
\circ j^1s
-\dfrac{\partial s^\gamma }
{\partial x^j}
\Bigl(
\dfrac{\partial ^3L}
{\partial y^\alpha \partial y^\gamma
\partial y_j^\sigma }
\circ j^1s
\Bigr)
\right. \\
& \! \!
\left.
-\dfrac{\partial ^2 s^\gamma }
{\partial x^h\partial x^j}
\Bigl(
\dfrac{\partial ^3L}
{\partial y^\alpha \partial y_h^\gamma
\partial y_j^\sigma }
\circ j^1s
\Bigr)
\right)
v^\alpha \\
& \!\! +\left(
\! \dfrac{\partial ^2 L}
{\partial y^\sigma
\partial y_i^\alpha }
\circ j^1 s
-\dfrac{\partial ^2 L}
{\partial y^\alpha
\partial y_i^\sigma }
\circ j^1 s
-\dfrac{\partial ^3L}
{\partial x^j\partial y_i^\alpha
\partial y_j^\sigma }
\circ j^1 s
\right. \\
& \!\!
\left.
-\dfrac{\partial s^\beta }
{\partial x^j}
\Bigl(
\dfrac{\partial ^3L}
{\partial y^\beta \partial y_i^\alpha
\partial y_j^\sigma }
\circ j^1s
\Bigr)
-\dfrac{\partial ^2 s^\beta }
{\partial x^h\partial x^j}
\Bigl(
\dfrac{\partial ^3L}
{\partial y_i^\alpha \partial y_h^\beta
\partial y_j^\sigma }
\circ j^1s
\Bigr)
\! \right)
\dfrac{\partial v^\alpha }
{\partial x^i} \\
& \!\!
-\Bigl(
\dfrac{\partial ^2 L}
{\partial y_j^\alpha \partial y_i^\sigma }
\circ j^1s
\Bigr)
\dfrac{\partial ^2 v^\alpha }
{\partial x^i\partial x^j}.
\end{array}
\end{equation*}
\end{remark}

\section{H-C equations
and gauge-invariance\label{HCeqs_GaugeInv}}

\subsection{Gauge-invariant
Lagrangians\label{Gauge_invariant_Lag}}

\begin{definition}
A smooth function $L\colon J^1C\to \mathbb{R}$
is said to be \emph{gauge invariant} if
$L\circ\Phi _C^{(1)}=L$ for every
$\Phi \in \mathrm{Gau}P$.
\end{definition}

This equation obviously implies $X_C^{(1)}L=0$
for every $X\in \mathrm{gau}P$.
The converse also holds if the group $G$
is connected. As every $\Phi \in \mathrm{Gau}P$
induces the identity map on $M$, the function $L$
is gauge invariant if and only if the gauge group
is a group of symmetries for the Lagrangian
density $\Lambda =L\mathbf{v}$, where $\mathbf{v}$
is an arbitrary volume form on the base manifold.

Let
\begin{equation}
\begin{array}
[c]{l}
\Omega \colon J^1C
\to \bigwedge \nolimits ^2T^\ast M
\otimes \mathrm{ad}P \\
\Omega (j_x^1\sigma _\Gamma )
=\left(
\Omega _\Gamma
\right) _x
\end{array}
\label{curvature}
\end{equation}
be the curvature map. Here, the curvature form
$\Omega _\Gamma $ of the connection $\Gamma $
corresponding to a section $s_\Gamma $
of $p$, is seen to be a two form on $M$
with values in the adjoint bundle $\mathrm{ad}P$.
On the vector bundle
$\bigwedge \nolimits ^2T^\ast M\otimes\mathrm{ad}P$
we consider the coordinate systems $(x^i;R_{jk}^\alpha )$,
$j<k$, induced by a coordinate system $(U;x^i)$ on $M$,
and a basis $(B_\alpha )$ of $\mathfrak{g}$, as follows:
\[
\eta _2=\sum _{j<k}
\left(
R_{jk}^\alpha (\eta _2)dx^j\wedge dx^k
\otimes\tilde{B}_\alpha
\right) _x,
\quad
\forall \eta _2\in \bigwedge \nolimits ^2T_x^\ast M
\otimes(\mathrm{ad}P)_x.
\]

The geometric formulation of Utiyama's Theorem
(e.g., see \cite{Bl}) states that a Lagrangian
$L\colon J^1C\to \mathbb{R}$ is gauge invariant
if and only $L$ factors through
$\Omega $ as $L=\tilde{L}\circ \Omega $, where
\begin{equation}
\tilde{L}\colon \bigwedge \nolimits ^2T^\ast M
\otimes \mathrm{ad}P
\to \mathbb{R} \label{Lbarra}
\end{equation}
is a differentiable function which is invariant under
the adjoint representation of $G$ on the curvature
bundle. As the curvature map \eqref{curvature}
is surjective, the function $\tilde{L}$ is unique.

\subsection{Projecting
$\mathcal{\bar{S}}_\Lambda $ onto
$\mathcal{S}_\Lambda $\label{projecting}}

From Lemma \ref{L1} in Section \ref{sectionHC},
we readily obtain

\begin{proposition}
\label{p1}
The H-C equations of a Lagrangian $L$ on the bundle
of connections $p\colon C\to M$ of a principal bundle
$\pi \colon P\to M$, read as follows:
\[
\bigl(
s_{h,j}^\beta -\bar{s}_{h,j}^\beta
\bigr)
\Bigl(
\dfrac{\partial ^2L}
{\partial A_{i,k}^\alpha \partial A_{h,j}^\beta }
\circ\bar{s}
\Bigr)
=0,\quad
\forall \alpha ,i,k,
\]
\[
-\dfrac{\partial }{\partial x^j}
\Bigl(
\dfrac{\partial L}{\partial A_{i,j}^\alpha }
\circ \bar{s}
\Bigr)
+\dfrac{\partial L}{\partial A_i^\alpha }
\circ \bar{s}
+\bigl(
s_{h,j}^\beta -\bar{s}_{h,j}^\beta
\bigr)
\Bigl(
\dfrac{\partial ^2L}
{\partial A_i^\alpha \partial A_{h,j}^\beta }
\circ \bar{s}
\Bigr)
=0,
\quad
\forall \alpha ,i,
\]
where $\bar{s}\colon M\to J^1C$ is a section
of $p_1\colon J^1C\to M$, and we have set
\[
s_i^\alpha
=A_i^\alpha \circ p_{10}\circ \bar{s},
\quad
s_{i,j}^\alpha
=\frac{\partial s_i^\alpha }{\partial x^j},
\quad
\bar{s}_{i,j}^\alpha
=A_{i,j}^\alpha \circ \bar{s}.
\]

\end{proposition}

\begin{lemma}
The Hessian metric of a gauge-invariant Lagrangian
$L$ on the bundle of connections $p\colon C\to M$
of a principal bundle $\pi\colon P\to M$, is singular;
i.e.,
\[
\det
\Bigl(
\frac{\partial ^2L}
{\partial A_{i,j}^\alpha \partial A_{k,l}^\beta }
\Bigr) _{\beta ,k,l}^{\alpha ,i,j}
=0.
\]
\end{lemma}

\begin{proof}
As $L=\tilde{L}\circ \Omega $, we have
$\partial L/\partial A_{i,i}^\alpha =0$,
taking the curvature equations into account,
i.e.,
\[
R_{ij}^\alpha \circ \Omega
=A_{i,j}^\alpha -A_{j,i}^\alpha
-c_{\beta \gamma }^\alpha A_i^\beta A_j^\gamma .
\]
Hence
$\partial ^2L/\partial A_{i,i}^\alpha
\partial A_{k,l}^\beta =0$,
for all indices $\beta,k,l$.
\end{proof}

\begin{definition}\label{regular}
A Lagrangian $L\in C^\infty (J^1C)$
is said to be \emph{regular}
if the Hessian metric $\mathrm{Hess}(L)$
is non-singular.

A gauge-invariant Lagrangian $L\in C^\infty(J^1C)$
is said to be \emph{weakly regular} if the Hessian
metric $\mathrm{Hess}(\tilde{L})$ of the function
in \eqref{Lbarra} associated to $L$ according
to Utiyama's theorem, is non-singular.
\end{definition}

In terms of the coordinate system
$(x^i;R_{jk}^\alpha )$, $j<k$,
on $\bigwedge ^2T^\ast M\otimes \mathrm{ad}P$
introduced in Section \ref{Gauge_invariant_Lag},
this means
\begin{equation}
\det
\Bigl(
\frac{\partial ^2\tilde{L}}
{\partial R_{ij}^\alpha \partial R_{kl}^\beta }
\Bigr) _{\beta ,k<l}^{\alpha ,i<j}\neq 0. \label{wc}
\end{equation}

\begin{remark}
The inequation \eqref{wc} imposes a generic condition
on $\tilde{L}$. In fact, most of gauge-invariant Lagrangians
in the field theory satisfy the weak regularity condition
\eqref{wc}; for example, the general Yang-Mills Lagrangian
on the bundle of connections of a principal bundle $P\to M$
with semisimple Lie group $G$ over a pseudo-Riemannian
manifold $(M,g)$ (even when constructed by using
a non-degenerate adjoint-invariant pairing on the Lie algebra
other than the Cartan-Killing pairing, see \cite{CM2, CMR})
is weakly regular. More generally, any quadratic function
as in \eqref{Lbarra}, which is simultaneously invariant
under the adjoint representation and under the action
of the gauge group of the principal bundle of $g$-orthonormal
linear frames, gives rise to a weakly regular Lagrangian,
see \cite{CM3}. Similarly, Born-Infeld Lagrangians (e.g.,
see \cite{Okawa, Suz}) are also weakly regular. We remark
on the fact that some special Lagrangians are not
weakly regular; for example, if we let the matter field
vanish in the Seiberg-Witten Lagrangian (e.g., see \cite{JPW})
then we obtain a Lagrangian on the bundle of connections,
which is not weakly regular (basically, because it involves
only the self-dual part of the curvature). Finally,
we should also remark that non-gauge invariant Lagrangians
(in the sense of Section \ref{Gauge_invariant_Lag})
may produce gauge invariant actions, as the Chern-Simons
Lagrangian. All of them are not considered below.
\end{remark}

Given an arbitrary fibred manifold $p\colon E\to M$,
we recall that the projection $p_{10}\colon J^1E\to E$
is endowed with an affine-bundle structure modelled over
$T^\ast M\otimes _EV(p)=p^\ast T^\ast M\otimes V(p)$;
see the item (b) in Section \ref{AffineBundles}.
In the particular case of the bundle of connections,
which is itself an affine bundle modelled over
$T^\ast M\otimes\mathrm{ad}P$, we conclude that
$p_{10}\colon J^1C\to  C$ is an affine bundle modelled
over
\begin{equation}
T^\ast M\otimes _CV(p)
=p^\ast (\otimes^2T^\ast M\otimes \mathrm{ad}P).
\label{T*M_otimes_V(p)}
\end{equation}
Hence, sections of $T^\ast M\otimes _CV(p)$
can be considered as $\mathrm{ad}P$-valued covariant
tensors of degree $2$ on $M$ with coefficients in $C$.

\begin{theorem}
\label{th1}
Let $\mathcal{\bar{S}}_\Lambda $
(resp.\ $\mathcal{S}_\Lambda $)
denote the set of solutions to H-C
(resp.\ E-L) equations of a weakly
regular gauge-invariant Lagrangian
$L$ on the bundle of connections
$p\colon C\to M$ of a principal bundle
$\pi \colon P\to M$.
If $\bar{s}\colon M\to J^1C$ belongs
to $\mathcal{\bar{S}}_\Lambda $,
then the section
$s=p_{10}\circ \bar{s}$ belongs
to $\mathcal{S}_\Lambda $.
Hence a natural projection exists
$\varrho \colon \mathcal{\bar{S}}_\Lambda
\to \mathcal{S}_\Lambda $,
$\varrho(\bar{s})=p_{10}\circ\bar{s}$,
which is an affine bundle modelled
as follows:
\[
\varrho ^{-1}(s)
=\left\{
j^1s+t:t\in \Gamma (S^2T^\ast M
\otimes \mathrm{ad}P)
\right\} ,
\quad
\forall s\in \mathcal{S}_\Lambda .
\]

\end{theorem}

\begin{proof}
We begin with the first H-C equation
in Proposition \ref{p1}.
As $L=\tilde{L}\circ \Omega $, we obtain
\begin{equation}
\frac{\partial L}
{\partial A_{i,i}^\alpha }=0,
\qquad
\frac{\partial L}
{\partial A_{i,k}^\alpha }
=\frac{\partial \tilde{L}}
{\partial R_{ik}^\alpha }
\circ \Omega ,\label{de}
\end{equation}
where we have set
$R_{ik}^\alpha =-R_{ki}^\alpha $
for $i>k$. Hence
\[
\bigl(
s_{h,j}^\beta -\bar{s}_{h,j}^\beta
\bigr)
\Bigl(
\dfrac{\partial ^2\tilde{L}}
{\partial R_{ik}^\alpha \partial R_{hj}^\beta }
\circ (\Omega \circ \bar{s})
\Bigr)
=0.
\]
If we assume the weak regularity condition
\eqref{wc} holds, then the previous equation
yields
\begin{equation}
\bigl(
s_{h,j}^\beta -\bar{s}_{h,j}^\beta
\bigr)
-\bigl(
s_{j,h}^\beta -\bar{s}_{j,h}^\beta
\bigr)
=0,\qquad \forall \beta ,h,j. \label{sim}
\end{equation}
If we write $\bar{s}=j^1s+t$, for
a $2$-tensor $t$, the condition above
means that $t$ is symmetric;
that is, $t$ is a section of
$S^2T^\ast M\otimes \mathrm{ad}P\to M$.

Next, we study the second equation
in Proposition \ref{p1}. Taking
the equations \eqref{de} into account,
we have
\[
-\dfrac{\partial }{\partial x^j}
\Bigl(
\dfrac{\partial L}
{\partial A_{i,j}^\alpha }
\circ \bar{s}
\Bigr)
+\dfrac{\partial L}{\partial A_i^\alpha }
\circ \bar{s}
+\bigl(
s_{h,j}^\beta -\bar{s}_{h,j}^\beta
\bigr)
\Bigl(
\frac{\partial }{\partial A_i^\alpha }
\Bigl(
\frac{\partial \tilde{L}}
{\partial R_{hj}^\beta }
\circ \Omega
\Bigr)
\circ \bar{s}
\Bigr)
=0.
\]
The last term vanishes identically
as $s_{h,j}^\beta -\bar{s}_{h,j}^\beta $
is symmetric by virtue of \eqref{sim}
and we have
$\partial\tilde{L}/\partial R_{hj}^\beta
=-\partial\tilde{L}/\partial R_{jh}^\beta $.
Then, the second equation reduces to
\begin{equation}
-\dfrac{\partial }{\partial x^j}
\Bigl(
\dfrac{\partial L}{\partial A_{i,j}^\alpha }
\circ\bar{s}
\Bigr)
+\dfrac{\partial L}{\partial A_i^\alpha }
\circ\bar{s}=0,
\label{pseudoE-L}
\end{equation}
which is precisely the E-L equation,
but evaluated at $\bar{s}$ instead of $j^1s$.
Nevertheless, the following formula
is readily checked:
\begin{equation}
\dfrac{\partial L}{\partial A_i^\alpha }=2
\Bigl(
c_{\alpha \gamma }^\beta A_j^\gamma
\frac{\partial \tilde{L}}{\partial R_{ij}^\beta }
\Bigr)
\circ \Omega . \label{A^alpha_i}
\end{equation}
Moreover, from \eqref{sim} we deduce
\begin{align}
\Omega \circ\bar{s}
& =(\bar{s}_{i,j}^\alpha
-\bar{s}_{j,i}^\alpha
-c_{\beta \gamma }^\alpha s_i^\beta s_j^\gamma )
dx^i\wedge dx^j\otimes\tilde{B}_\alpha
\label{curr} \\
& =(s_{i,j}^\alpha
-s_{j,i}^\alpha
-c_{\beta \gamma }^\alpha s_i^\beta s_j^\gamma )
dx^i\wedge dx^j\otimes \tilde{B}_\alpha
\nonumber \\
& =\Omega \circ j^1s. \nonumber
\end{align}
Therefore, from the formulas \eqref{de}
and \eqref{A^alpha_i} we conclude that the equation
\eqref{pseudoE-L} coincides with the E-L equation
for $s$.
\end{proof}

\begin{remark}
As metioned in the introduction,
the formula \eqref{curr} shows that
the curvature remains constant along
the fibre of $\varrho$ over any
$j^1s\in \mathcal{S}_\Lambda $.
\end{remark}

\begin{corollary}
\label{corol_0}
The H-C equations of a weakly-regular
gauge-invariant Lagran\-gian $L$
on the bundle of connections
$p\colon C\to M$ of a principal bundle
$\pi\colon P\to M$, are equivalent
to the following system:
\[
\bigl(
s_{h,j}^\alpha -\bar{s}_{h,j}^\alpha
\bigr)
-\bigl(
s_{j,h}^\alpha -\bar{s}_{j,h}^\alpha
\bigr)
=0,
\qquad
\forall \alpha ,h,j.
\]
\[
-\dfrac{\partial }{\partial x^j}
\Bigl(
\dfrac{\partial L}{\partial A_{i,j}^\alpha }
\circ j^1s
\Bigr)
+\dfrac{\partial L}{\partial A_i^\alpha }
\circ j^1s=0,
\quad
\forall \alpha ,i,
\]
where $\bar{s}\colon M\to  J^1C$ is a section
of $p_1\colon J^1C\to M$, and $s_i^\alpha $,
$s_{i,j}^\alpha $, and $\bar{s}_{i,j}^\alpha $
are as in \emph{Proposition \ref{p1}}.
\end{corollary}

\begin{remark}
\label{remark}
For every section $\bar{s}\colon M\to  J^1C$
of $p_1\colon J^1C\to M$, let $s\colon M\to C$
be the section of $p\colon C\to M$ defined by
$s=p_{10}\circ\bar{s}$. As the points
$\bar {s}(x),j_x^1s\in J^1C$ lie over the same
fibre of $p_{10}\colon J^1C\to C$ and this map
admits an affine-bundle structure modelled over
the vector bundle \eqref{T*M_otimes_V(p)}, a map
\[
\delta _{C}\colon J^1(p_1)\to \otimes ^2T^\ast M
\otimes \mathrm{ad}P
\]
exists such that,
$\delta _C(j_x^1\bar{s})=\bar{s}(x)-j_x^1s$.
If $\mathrm{alt}\colon \otimes ^2T^\ast M
\otimes \mathrm{ad}P\to \bigwedge \nolimits^2T^\ast M
\otimes \mathrm{ad}P$ denotes the anti-symmetrization
operator, then the first group of H-C equations
in Corollary \ref{corol_0} means that $j^1\bar{s}$
takes values into the subbundle
$\ker (\mathrm{alt}\circ \delta _C)$.
\end{remark}

\begin{proposition}
\label{L3}
If $\Lambda $ is a gauge-invariant Lagrangian
density on the bundle of connections
of a principal bundle
$\pi \colon P\to M$ and
$\bar{s}\in \mathcal{\bar{S}}_\Lambda $
(resp.\ $s\in \mathcal{S}_\Lambda $), then
$\Phi _C^{(1)}\circ \bar{s}
\in \mathcal{\bar{S}}_\Lambda $
(resp.\ $\Phi _C\circ s\in \mathcal{S}_\Lambda $)
for every $\Phi \in \mathrm{Gau}P$. Accordingly,
the gauge group of $P$ acts (on the left)
on $\mathcal{\bar{S}}_\Lambda $ (resp.\
$\mathcal{S}_\Lambda $) by setting
$\Phi \cdot\bar{s}=\Phi _C^{(1)}\circ\bar{s}$
(resp.\ $\Phi \cdot s=\Phi _{C}\circ s$),
$\forall \bar{s}\in \mathcal{\bar{S}}_\Lambda $
(resp.\ $\forall s\in \mathcal{S}_\Lambda $),
$\forall \Phi \in \mathrm{Gau}P$.
\end{proposition}

\begin{proof}
We prove that the section
$\Phi _C^{(1)}\circ \bar{s}$
of $p_1$ is a solution to H-C equation
\eqref{HC}. For every
$Z\in \mathfrak{X}^v(J^1E)$, we set
\[
Y=\left(
\Phi _C^{(1)}
\right) ^{-1}
\!\! \cdot Z\in \mathfrak{X}^v(J^1E).
\]
As $\Lambda $ is gauge invariant, we have
$(\Phi _{C}^{(1)})^\ast \Lambda =\Lambda $,
from the functorial character
of the Poincar\'e-Cartan form
(see \cite{Garcia0}) we obtain
\begin{align*}
\left(
\Phi _C^{(1)}\circ \bar{s}
\right) ^\ast
\left(
i_Zd\Theta _\Lambda
\right)
& =\bar{s}^\ast (\Phi _C^{(1)})^\ast
\left(
i_Zd\Theta _\Lambda
\right) \\
& =\bar{s}^\ast i_{Y}d
\left(
(\Phi _C^{(1)})^\ast \Theta _\Lambda
\right) \\
& =\bar{s}^\ast i_Yd
\Theta _{(\Phi _C^{(1)})^\ast \Lambda } \\
& =\bar{s}^\ast i_Yd
\Theta _\Lambda \\
& =0.
\end{align*}
The proof for the solutions to E-L equations,
is similar and therefore it is omitted.
\end{proof}

Theorem \ref{th1} shows that the set of solutions
of H-C trivially fibers over the set of solutions
of E-L. On the other hand, for gauge-invariant
problems, the moduli space of solutions under
the action of the (restricted) gauge group plays
a relevant role. We now study the relationship
between the moduli of H-C and E-L showing
that the first fibers over the second,
but not necessarily in a trivial way. First,
note that the gauge group $\mathrm{Gau}P$ acts
on the adjoint bundle $\mathrm{ad}P$ by setting
$\Phi _{\mathrm{ad}}((u,B)_{\mathrm{ad}})
=(\Phi (u),B)_{\mathrm{ad}}$,
$\forall \Phi \in \mathrm{Gau}P$,
$\forall (u,B)\in P\times\mathfrak{g}$,
and this action obviously induces another action
on $S^2T^\ast M\otimes \mathrm{ad}P$ as follows:
\begin{equation}
\begin{array}
[c]{l}
\Phi _{\mathrm{ad}P}
\left(
w_1\odot w_2\otimes v
\right)
=w_1\odot w_2\otimes \Phi _{\mathrm{ad}P}(v), \\
\forall w_1,w_2\in T_x^\ast M,
\;
\forall v\in
\left(
\mathrm{ad}P
\right) _x,
\end{array}
\label{Phi_adjoint}
\end{equation}
where the symbol $\odot $ denotes symmetric
product.

\begin{proposition}
\label{moduli}
Given a point $x_0\in M$, let
$\mathrm{Gau}_{x_0}P$ be the subgroup
of gauge transformations
$\Phi \in \mathrm{Gau}P$ such that,
$\Phi (u)=u$,
$\forall u\in \pi^{-1}(x_0)$. Then

\begin{enumerate}
\item[\emph{(i)}]
For every gauge-invariant Lagrangian density
$\Lambda $ on the bundle of connections
of $\pi \colon P\to M$, the quotient map
$\kappa _{P}\colon\mathcal{S}_\Lambda
\to \mathcal{S}_\Lambda /\mathrm{Gau}_{x_0}P$
is a set-theoretical principal
$\mathrm{Gau}_{x_0}P$-bundle.

\item[\emph{(ii)}]
In addition, if $\Lambda $ is weakly regular,
the projection
$\varrho \colon \mathcal{\bar{S}}_\Lambda
\to \mathcal{S}_\Lambda $
defined in \emph{Theorem \ref{th1}},
induces a mapping
\[
\begin{array}
[c]{l}
\varrho _{\mathrm{Gau}_{x_0}P}
\colon \mathcal{\bar{S}}_\Lambda
/\mathrm{Gau}_{x_0}P
\to \mathcal{S}_\Lambda
/\mathrm{Gau}_{x_0}P, \\
\varrho _{\mathrm{Gau}_{x_0}P}
\left(
\bar{s}\operatorname{mod}
\mathrm{Gau}_{x_0}P
\right)
=\varrho (\bar{s})\operatorname{mod}
\mathrm{Gau}_{x_0}P,
\end{array}
\]
which is the vector bundle associated
to the principal bundle $\kappa _P$
by the action on
$S^2\Omega ^1(M)\otimes \mathrm{gau}P$
induced on the sections
of $S^2T^\ast M\otimes \mathrm{ad}P$
by the action of $\mathrm{Gau}P$
defined in the formula
\emph{\eqref{Phi_adjoint}} above.
\end{enumerate}
\end{proposition}

\begin{proof}
As is known, $\mathrm{Gau}_{x_0}P$
acts freely on the space of connections,
i.e., on the sections of $p\colon C\to M$
and, in particular, on
$\mathcal{\bar{S}}_\Lambda $
and on $\mathcal{S}_\Lambda $ (e.g., see
\cite[Theorem 2.2.4]{CR}, \cite[III.C]{MV}).
Nevertheless,
the quotients
$\mathcal{\bar{S}}_\Lambda /\mathrm{Gau}_{x_0}P$
and $\mathcal{S}_\Lambda /\mathrm{Gau}_{x_0}P$
may be singular, e.g., see \cite[p.\ 134]{CR}.
Because of this, we consider such structure
from the set-theoretical point of view only.

Furthermore, the mapping
$\varrho _{\mathrm{Gau}_{x_0}P}$
is well defined as
\begin{align*}
\varrho
\left(
\Phi _{C}^{(1)}\circ\bar{s}
\right)
& =\left(
p_{10}\circ
\Phi _C^{(1)}
\right)
\circ \bar{s} \\
& =\left(
\Phi _C\circ p_{10}
\right)
\circ\bar{s} \\
& =\Phi _C\circ \varrho (\bar{s}),
\end{align*}
and from Theorem \ref{th1} it follows
that every $\bar{s}
\in \mathcal{\bar{S}}_\Lambda $
can be uniquely written as
$\bar{s}=j^1s+t$, where
$s=\varrho (\bar{s})$ and
$t\in S^2\Omega ^1(M)
\otimes \mathrm{gau}P$.
Hence $\bar{s}$ can be identified
to the pair $(s,t)$,
i.e.,
$\mathcal{\bar{S}}_\Lambda
\cong \mathcal{S}_\Lambda
\times S^2\Omega ^1(M)
\otimes \mathrm{gau}P$.
Recalling that $\Phi _C\colon C
\to C$ is an affine-bundle
morphism whose associated
vector-bundle is
$\mathrm{id}_{T^\ast M}
\otimes \Phi _{\mathrm{ad}}
\colon T^\ast M\otimes \mathrm{ad}P
\to T^\ast M\otimes \mathrm{ad}P$,
we have
\begin{align*}
\Phi _C^{(1)}\circ \bar{s}
& =\Phi _C^{(1)}\circ
\left(
j^1s+t
\right) \\
& =J^1
\left(
\Phi _C\circ s
\right)
+\Phi _{\mathrm{ad}}\cdot t,
\end{align*}
thus concluding the proof.
\end{proof}

\begin{theorem}
\label{T2}
With the same notations as
in \emph{Section \ref{sectionJacobi}}
and the same assumptions as
in \emph{Theorem \ref{th1}}, if
$\bar{X}\in T_{\bar{s}}
\mathcal{\bar{S}}_\Lambda $,
then
$(p_{10})_\ast \circ \bar{X}
\in T_s\mathcal{S}_\Lambda $,
where $s=p_{10}\circ\bar{s}$. Hence,
the natural map $\varrho \colon
\mathcal{\bar{S}}_\Lambda
\to \mathcal{S}_\Lambda $ induces
a linear map
$\varrho _\ast \colon T_{\bar{s}}
\mathcal{\bar{S}}_\Lambda
\to T_s\mathcal{S}_\Lambda $,
$\varrho _\ast(\bar{X})
=(p_{10})_\ast \circ \bar{X}$.
Moreover, $\ker \varrho _\ast
\cong \Gamma
(S^2T^\ast M\otimes \mathrm{ad}P)$.
\end{theorem}

\begin{proof}
We first begin with the second Jacobi
equation \eqref{Jacobi_II}
for a gauge-invariant Lagrangian
$L\colon J^1C\to \mathbb{R}$
and a Jacobi vector field
\[
\bar{X}=v_i^\alpha
\frac{\partial }
{\partial A_i^\alpha }
+v_{i,j}^\alpha
\frac{\partial }
{\partial A_{i,j}^\alpha },
\quad
v_i^\alpha ,v_{i,j}^\alpha
\in C^\infty (M),
\]
along a solution $\bar{s}
\in \mathcal{\bar{S}}_\Lambda $.
We have
\begin{equation}
\begin{array}
[c]{cl}
0= & \Bigl(
\dfrac{\partial ^2L}
{\partial A_{r,i}^\alpha
\partial A_{h,j}^\sigma }
\circ \bar{s}
\Bigr)
\Bigl(
\dfrac{\partial v_r^\alpha }
{\partial x^i}
-v_{r,i}^\alpha
\Bigr)
\medskip \\
& +\Bigl(
\dfrac{\partial s_r^\alpha }
{\partial x^i}-\bar{s}_{r,i}^\alpha
\Bigr)
\Bigl(
\dfrac{\partial ^3L}
{\partial A_l^\beta
\partial A_{r,i}^\alpha
\partial A_{h,j}^\sigma }
\circ\bar{s}
\Bigr)
v_l^\beta
\medskip \\
& +\Bigl(
\dfrac{\partial s_r^\alpha }
{\partial x^i}
-\bar{s}_{r,i}^\alpha
\Bigr)
\Bigl(
\dfrac{\partial ^3L}
{\partial A_{r,i}^\alpha
\partial A_{l,k}^\beta
\partial A_{h,j}^\sigma }
\circ\bar{s}
\Bigr)
v_{l,k}^\beta ,
\end{array}
\label{GaugeJacobi_II}
\end{equation}
for any $\sigma ,h,j$. From Theorem
\ref{th1}, $j^1s-\bar{s}$ is a symmetric
tensor. Moreover, taking the formula
\eqref{de} into account, the last two
summands of \eqref{GaugeJacobi_II}
vanish. We thus obtain
\[
0=\Bigl(
\dfrac{\partial ^2L}
{\partial A_{r,i}^\alpha
\partial A_{s,j}^\sigma }
\circ\bar{s}
\Bigr)
\Bigl(
\dfrac{\partial v_r^\alpha }
{\partial x^i}
-v_{r,i}^\alpha
\Bigr),
\]
which, assuming the weak regularity of $L$,
implies
\begin{equation}
\label{tensor}
v_{r,i}^\alpha
=\frac{\partial v_r^\alpha }
{\partial x^i}
+t_{r,i}^\alpha ,
\end{equation}
where $t_{r,i}^\alpha $ are the components
of a symmetric tensor $t\in \Gamma
(S^2T^\ast M \otimes \mathrm{ad}P)$.
From the formula \eqref{i.c.t.} and Lemma \ref{l1}
we conclude $\bar{X}=X^{(1)}+t$,
where $X^{(1)}$ is the $1$-jet prolongation
of the vector field along $s$ given by
$X=v_i^\alpha \partial /\partial A_i^\alpha $.

Next, we consider the first Jacobi equation
\ref{Jacobi_I}, from which we obtain

{\small
\begin{equation}
\begin{array}
[c]{cl}
0= &
\left(
\dfrac{\partial ^2L}
{\partial A_r^\alpha
\partial A_q^\sigma }
\circ\bar{s}
-\dfrac{\partial ^3L}
{\partial x^i\partial A_r^\alpha
\partial A_{q,i}^\sigma }
\circ\bar{s}
+\Bigl(
s_{t,i}^\beta -\bar{s}_{t,i}^\beta
\Bigr)
\Bigl(
\dfrac{\partial ^3L}
{\partial A_r^\alpha
\partial A_q^\sigma
\partial A_{t,i}^\beta }
\circ\bar{s}
\Bigr)
\right. \\
& \left.
-s_{t,j}^\gamma
\Bigl(
\dfrac{\partial ^3L}
{\partial A_r^\alpha
\partial A_t^\gamma
\partial A_{q,j}^\sigma }
\circ\bar{s}
\Bigr)
-\dfrac{\partial \bar{s}_{t,h}^\gamma }
{\partial x^j}
\Bigl(
\dfrac{\partial ^3L}
{\partial A_r^\alpha
\partial A_{t,h}^\gamma
\partial A_{q,j}^\sigma }
\circ\bar{s}
\Bigr)
\right)
v_r^\alpha \\
&
+\Bigl(
\dfrac{\partial ^2L}
{\partial A_q^\sigma
\partial A_{r,i}^\alpha }
\circ \bar{s}
-\dfrac{\partial ^2L}
{\partial A_r^\alpha
\partial A_{q,i}^\sigma }
\circ \bar{s}
\Bigr)
\dfrac{\partial v_r^\alpha }
{\partial x^i}
-\Bigl(
\dfrac{\partial ^2L}
{\partial A_{r,j}^\alpha
\partial A_{q,i}^\sigma }
\circ\bar{s}
\Bigr)
\dfrac{\partial v_{r,j}^\alpha }
{\partial x^i} \\
& +\left(
\Bigl(
s_{t,j}^\beta -\bar{s}_{t,j}^\beta
\Bigr)
\Bigl(
\dfrac{\partial ^3L}
{\partial A_q^\sigma
\partial A_{r,i}^\alpha
\partial A_{t,j}^\beta }
\circ\bar{s}
\Bigr)
-\dfrac{\partial ^3L}
{\partial x^j\partial A_{r,i}^\alpha
\partial A_{q,j}^\sigma }
\circ\bar{s}
\right. \\
& \left.
-s_{t,j}^\beta
\Bigl(
\dfrac{\partial ^3L}
{\partial A_t^\beta
\partial A_{r,i}^\alpha
\partial A_{q,j}^\sigma }
\circ \bar{s}
\Bigr)
-\dfrac{\partial \bar{s}_{t,h}^\beta }
{\partial x^j}
\Bigl(
\dfrac{\partial ^3L}
{\partial A_{t,h}^\beta
\partial A_{r,i}^\alpha
\partial A_j^\sigma }
\circ\bar{s}
\Bigr)
\right)
v_{r,i}^\alpha .
\end{array}
\label{Jacobi_I_con}
\end{equation}
}
Again taking the symmetry of the differences
$s^\beta _{t,i}-\bar{s}^\beta _{t,i}$
(and their derivatives) into account,
the previous equation reduces to the following:

\begin{equation}
\begin{array}
[c]{cl} 0
= & \left(
\dfrac{\partial ^2L}
{\partial A_r^\alpha
\partial A_q^\sigma }
\circ \bar{s}
-\dfrac{\partial ^3L}
{\partial x^i\partial A_r^\alpha
\partial A_{q,i}^\sigma }
\circ \bar{s}-s_{t,j}^\gamma
\Bigl(
\dfrac{\partial ^3L}
{\partial A_r^\alpha
\partial A_t^\gamma
\partial A_{q,j}^\sigma }
\circ\bar{s}
\Bigr)
\right. \\
& \left.
-\dfrac{\partial ^2s_t^\gamma }
{\partial x^h\partial x^j}
\Bigl(
\dfrac{\partial ^3L}
{\partial A_r^\alpha
\partial A_{t,h}^\gamma
\partial A_{q,j}^\sigma }
\circ \bar{s}
\Bigr)
\right)
v_r^\alpha \\
& +\Bigl(
\dfrac{\partial ^2L}
{\partial A_q^\sigma
\partial A_{r,i}^\alpha }
\circ \bar{s}
-\dfrac{\partial ^2L}
{\partial A_r^\alpha
\partial A_{q,i}^\sigma }
\circ \bar{s}
\Bigr)
\dfrac{\partial v_r^\alpha }
{\partial x^i}
-\Bigl(
\dfrac{\partial ^2L}
{\partial A_{r,j}^\alpha
\partial A_{q,i}^\sigma }
\circ\bar{s}
\Bigr)
\dfrac{\partial v_{r,j}^\alpha }
{\partial x^i} \\
& -\left(
\dfrac{\partial ^3L}
{\partial x^j\partial A_{r,i}^\alpha
\partial A_{q,j}^\sigma }
\circ \bar{s}+s_{t,j}^\beta
\Bigl(
\dfrac{\partial ^3L}
{\partial A_t^\beta
\partial A_{r,i}^\alpha
\partial A_{q,j}^\sigma }
\circ \bar{s}
\Bigr)
\right. \\
& \left.
+\dfrac{\partial ^2 s_t ^\beta}
{\partial x^h\partial x^j}
\Bigl(
\dfrac{\partial ^3L}
{\partial A_{t,h}^\beta
\partial A_{r,i}^\alpha
\partial A_{q,j}^\sigma }
\circ\bar{s}
\Bigr)
\right)
v_{r,i}^\alpha .
\end{array}
\label{Jacobi_I_con_bis}
\end{equation}

By taking derivatives in \eqref{tensor},
we obtain
\[
\frac{\partial v_{r,j}^\alpha }
{\partial x^i}
=\frac{\partial v_r^\alpha }
{\partial x^j\partial x^i}
+\frac{\partial t_{r,j}^\alpha }
{\partial x^i},
\]
and substituting these expressions
into the equation \eqref{Jacobi_I_con_bis},
again by virtue of the skew-symmetry
of $\partial L/\partial A^\alpha _{i,j}$,
we have
\[
\begin{array}
[c]{cl}
0= & \left(
\dfrac{\partial ^2L}
{\partial A_r^\alpha
\partial A_q^\sigma }
\circ \bar{s}
-\dfrac{\partial ^3L}
{\partial x^i\partial A_r^\alpha
\partial A_{q,i}^\sigma }
\circ \bar{s}
-s_{t,j}^\gamma
\Bigl(
\dfrac{\partial ^3L}
{\partial A_r^\alpha
\partial A_t^\gamma
\partial A_{q,j}^\sigma }
\circ\bar{s}
\Bigr)
\right. \\
& \left.
-\dfrac{\partial ^2s_t^\gamma }
{\partial x^j\partial x^h}
\Bigl(
\dfrac{\partial ^3L}
{\partial A_r^\alpha
\partial A_{t,h}^\gamma
\partial A_{q,j}^\sigma }
\circ\bar{s}
\Bigr)
\right)
v_r^\alpha \\
& +\Bigl(
\dfrac{\partial ^2L}
{\partial A_q^\sigma
\partial A_{r,i}^\alpha }
\circ \bar{s}
-\dfrac{\partial ^2L}
{\partial A_r^\alpha
\partial A_{q,i}^\sigma }
\circ \bar{s}
\Bigr)
\dfrac{\partial v_r^\alpha }
{\partial x^i}
-\Bigl(
\dfrac{\partial ^2L}
{\partial A_{r,j}^\alpha
\partial A_{q,i}^\sigma }
\circ\bar{s}
\Bigr)
\dfrac{\partial v_r^\alpha }
{\partial x^i\partial x^j} \\
& -\left(
\dfrac{\partial ^3L}
{\partial x^j\partial A_{r,i}^\alpha
\partial A_{q,j}^\sigma }
\circ \bar{s}+s_{t,j}^\beta
\Bigl(
\dfrac{\partial ^3L}
{\partial A_t^\beta
\partial A_{r,i}^\alpha
\partial A_{q,j}^\sigma }
\circ \bar{s}
\Bigr)
\right. \\
& \left.
+\dfrac{\partial ^2 s_t^\beta }
{\partial x^h \partial x^j}
\Bigl(
\dfrac{\partial ^3L}
{\partial A_{t,h}^\beta
\partial A_{r,i}^\alpha
\partial A_{q,j}^\sigma }
\circ\bar{s}
\Bigr)
\right)
\dfrac{\partial v_r^\alpha }
{\partial x^i}.
\end{array}
\]
Finally, taking the equations
$\Omega \circ \bar{s}
=\Omega \circ j^1s$,
$L=\bar{L}\circ \Omega$,
into account, we conclude that
$\bar{s}$ can be replaced
by $j^1s$ into the previous equation
and we can end the proof by simply
applying Remark \ref{remm}.
\end{proof}

\begin{corollary}
A Jacobi field
$\bar{X}
\in T_{\bar{s}}
\bar{\mathcal{S}}_\Lambda $
is integrable if and only if
the Jacobi field
$\varrho _\ast (\bar{X})
=X\in T_s\mathcal{S}_\Lambda $,
$s=p_{10}\circ \bar{s}$,
is integrable.
\end{corollary}

\begin{proof}
If $\bar{X}
\in T_{\bar{s}}
\mathcal{\bar{S}}_\Lambda $
is an integrable Jacobi field, then
$\bar{X}
=d/d\varepsilon |_{\varepsilon =0}
\bar{s}_\varepsilon $
where
$\bar{s}_\varepsilon
=j^1s_\varepsilon +t_\varepsilon $,
with
$\pi _{10}\circ \bar{s}
=s_\varepsilon \in \mathcal{S}_\Lambda $,
$t_\varepsilon $ being a symmetric tensor.
Then,
\begin{equation*}
\varrho _\ast (\bar{X})
=(\pi _{10})_\ast \bar{X}
=\left.
\frac{d}{d\varepsilon }
\right| _{\varepsilon =0}
\pi _{10}\circ
\bar{s}_\varepsilon
=\left.
\frac{d}{d\varepsilon }
\right| _{\varepsilon =0}
s_\varepsilon .
\end{equation*}
Hence $X=\varrho _\ast (\bar{X})$
is integrable. Conversely, assume
$X=d/d\varepsilon |_{\varepsilon =0}
s_\varepsilon $,
$s_\epsilon \in \mathcal{S}_\Lambda $,
and
$\bar{X}\in (\varrho _\ast )^{-1}X $.
We know that
$\bar{X}=X^{(1)}+t$,
where $t$ is a symmetric tensor.
Then,
$\bar{X}
=d/d\varepsilon |_{\varepsilon =0}
\bar{s}_\varepsilon $ with
$\bar{s}_\varepsilon
=j^1s_\varepsilon +\varepsilon
t\in \mathcal{\bar{S}}_\Lambda $,
and $\bar{X}$ is integrable.
\end{proof}

\section{H-C self-dual and anti-self-dual connections}

Let $(M,g)$ be a pseudo-Riemannian $n$-dimensional
oriented connected manifold of signature $(n^+,n^-)$,
$n=n^++n^-$, and let $\mathbf{v}_g=\sqrt{|\det (g_{ij})|}
dx^1 \wedge \ldots \wedge dx^n$, $g_{ij}
=g(\partial /\partial x^i,\partial/\partial x^j)$,
be its pseudo-Riemannian volume form. The canonical
duality isomorphism attached to $g$ is denoted by
$T_xM\to T_x^\ast M$, $X\mapsto X^\flat $,
with inverse map $T_x^\ast M\to T_xM$,
$w\mapsto w^\sharp $. Let $g^{(r)}$ be the metric
on $\bigwedge ^rT^\ast M$ given by
$g^{(r)}(w^1\wedge \ldots \wedge w^r,\bar{w}^1
\wedge \ldots \wedge \bar{w}^r)
=\det (g((w^i)^\sharp ,(\bar{w}^j)^\sharp ))_{i,j=1}^r$.

Let $V\to M$ be a vector bundle.
The Hodge star can be extended
to $V$-valued forms as follows:
$\star (\omega _r\otimes v)
=(\star \omega _r)\otimes v$,
$\forall \omega _r
\in \bigwedge ^r T_x^\ast M$,
$\forall v\in V_x$.

Let $\pi\colon P\to M$ be a principal
$G$-bundle and let $\pi _{\mathrm{ad}P}
\colon \mathrm{ad}P\to M$
be the adjoint bundle; i.e.,
the bundle associated with $P$
under the adjoint representation of $G$
on its Lie algebra $\mathfrak{g}$.
For every $B\in \mathfrak{g}$ and every
$u\in P$, let $(u,B)_G$ be the coset
of $(u,B)\in P\times \mathfrak{g}$
modulo $G$. A symmetric bilinear form
$\langle \cdot,\cdot\rangle
\in S^2\mathfrak{g}^\ast $ is said to be
invariant under the adjoint representation
if the following equation holds:
$\left\langle \mathrm{Ad}_{g}B,
\mathrm{Ad}_gC\right\rangle
=\left\langle B,C\right\rangle $,
$\forall g\in G$,
$\forall B,C\in\mathfrak{g}$.
By taking derivatives
on this equation we obtain
$\left\langle \left[ A,B\right] ,
C\right\rangle +\left\langle B,
\left[ A,C\right] \right\rangle =0$,
$\forall A,B,C\in \mathfrak{g}$.
If the group $G$ is connected,
then both equations above are equivalent.

Every symmetric bilinear form
$\langle\cdot ,
\cdot \rangle\in S^2\mathfrak{g}^\ast $
invariant under the adjoint representation
induces a fibred metric
$\left\langle \! \left\langle \cdot ,
\cdot \right\rangle \!\right\rangle
\colon \mathrm{ad}P\oplus \mathrm{ad}P
\to \mathbb{R}$ by setting
\begin{equation}
\left\langle \! \left\langle (u,B)_G,
(u,C)_G\right\rangle \! \right\rangle
=\left\langle B,C\right\rangle ,
\quad
\forall u\in P,\;\forall B,C\in \mathfrak{g}.
\label{<< >>}
\end{equation}
We further assume that the pairing \eqref{<< >>}
is non-degenerate.

Every pseudo-Riemannian metric $g$ on $M$
and every fibred
$\left\langle \!\left\langle \cdot ,
\cdot\right\rangle \! \right\rangle $
on $\mathrm{ad}P$ induce a fibred metric
on the vector bundle of $\mathrm{ad}P$-valued
differential $r$-forms on $M$ as follows:
$(\! (\alpha _r\otimes a,\beta _r\otimes b
)\! )=g^{(r)}(\alpha _r,\beta _r)
\left\langle \! \left\langle
a,b\right\rangle \! \right\rangle $,
$\forall \alpha _r,\beta _r
\in \bigwedge ^rT_x^\ast M$,
and $\forall a,b\in (\mathrm{ad}P)_x$.
Moreover, the pairing \eqref{<< >>}
defines an exterior product (see \cite{Bl}),
\[
\begin{array}
[c]{l}
\dot{\wedge}\colon
\left(
\bigwedge ^\bullet
T^\ast M\otimes \mathrm{ad}P
\right)
\oplus
\left(
\bigwedge ^\bullet T^\ast M
\otimes \mathrm{ad}P
\right)
\to \bigwedge ^\bullet T^\ast M,\\
\left(
\alpha _q\otimes a
\right)
\dot{\wedge }
\left(
\beta _r\otimes
b\right)
=\left(
\alpha _q\wedge \beta _r
\right)
\left\langle
\!\left\langle
a,b
\right\rangle \!
\right\rangle .
\end{array}
\]

Let $p\colon C\to M$ be the bundle of connections
of $P$. According to the previous definitions,
a pseudo-Riemannian metric $g$ on $M$ and
an adjoint-invariant symmetric bilinear
form $\langle \cdot,\cdot \rangle $ allow one
to define a quadratic Lagrangian density
$\Lambda =L\mathbf{v}_g$ on $J^1C$ by setting,
\begin{align}
\Lambda
\left(
j_x^1s
\right)
& =\left( \! \left(
\Omega ^s(x),\Omega ^s(x)
\right) \! \right)
\mathbf{v}_g(x)
\label{YM1}\\
& =\Omega ^s(x)\dot{\wedge}\star \Omega ^s(x),
\nonumber
\end{align}
where $s$ is a local section of $p$ defining
a principal connection whose curvature form
is denoted by $\Omega ^s$.
In \cite{CM2} it is proved that the E-L equations
of the Lagrangians \ above are seen to be independent
of the pairing \eqref{<< >>} and they coincide
with the classical Yang-Mills equations:
 $\nabla ^s\star \Omega ^s=0$.

\begin{theorem}
\label{th}In addition to the hypotheses above,
assume $\dim M=4$ and $g$ is a Riemannian metric.
Let $\mathcal{S}_\Lambda ^+$ (resp.\
$\mathcal{S}_\Lambda ^-$) be the set of self-dual
(resp.\ anti-self-dual) connections
with respect to the Lagrangian density
\emph{\eqref{YM1}}.
A section $\bar{s}$ of $p_1\colon J^1C\to M$
belongs to the fibre $\varrho ^{-1}(s)$,
with $s\in\mathcal{S}_\Lambda ^+$ (resp.\
$s\in\mathcal{S}_\Lambda ^-$) if and only if
the following equations hold:
\begin{align}
\mathrm{alt}(\bar{s}-j^1s)
& =0,\label{eq1}\\
\star (\Omega \circ \bar{s})
& =(\Omega \circ \bar{s})
\quad
\text{(resp.\ }
\star(\Omega \circ \bar{s})
=-(\Omega \circ \bar{s})
\text{)}, \label{eq2}
\end{align}
where $\mathrm{alt}\colon \otimes ^2T^\ast M
\otimes \mathrm{ad}P\to \bigwedge ^2T^\ast M
\otimes\mathrm{ad}P$ is the alternating operator.
\end{theorem}

\begin{remark}
The equation \eqref{eq1} is a first-order differential
equation, whereas \eqref{eq2} is a purely algebraic
equation. The conditon $\mathrm{alt}(\bar{s}-j^1s)=0$
is not specific of the Yang-Mills Lagrangian
but general for any weakly regular gauge-invariant
Lagrangian. In fact, it defines the subset
of $\Gamma (p_1)$ given by
\begin{equation}
\{ j^1s+t:s\in \Gamma (p),t\in \Gamma
(S^2T^\ast M\otimes \mathrm{ad}P)\} .
\label{sym}
\end{equation}
It is thus interesting to note that the group
of equations really defined by self-dual (resp.\
anti-self-dual) connections, is not longer
a differential equation but an algebraic
constrain on the subset \eqref{sym}.
\end{remark}

\begin{proof}
[Proof of Theorem \ref{th}]
As $L$ is weakly regular, if
$\bar{s}\in \mathcal{\bar{S}}_\Lambda $,
then $\mathrm{alt}(\bar{s}-j^1s)=0$
by virtue of Remark \ref{remark}.
Moreover, as
\[
-\dfrac{\partial }{\partial x^j}
\Bigl(
\dfrac{\partial L}{\partial A_{i,j}^\alpha }
\circ \bar{s}
\Bigr)
+\dfrac{\partial L}{\partial A_i^\alpha }
\circ \bar{s}
=-\dfrac{\partial }{\partial x^j}
\Bigl(
\dfrac{\partial L}{\partial A_{i,j}^\alpha }
\circ j^1s
\Bigr)
+\dfrac{\partial L}{\partial A_i^\alpha }
\circ j^1s,
\]
as $\Omega \circ\bar{s}=\Omega \circ j^1s$,
we also obtain $\nabla ^s\star \Omega ^s=0$,
where $\Omega ^s=\Omega \circ j^1s$ and
$s=p_{10}\circ \bar{s}$. Hence, the definition
of a self-dual connection (i.e.,
$\star (\Omega \circ j^1s)=(\Omega \circ j^1s)$)
can be written as $\star (\Omega \circ \bar{s})
=(\Omega \circ \bar{s})$ for every
$\bar{s}\in \varrho ^{-1}(\mathcal{S}_\Lambda ^+)$.
Similarly for $\mathcal{S}_\Lambda ^-$.
\end{proof}

\end{document}